\documentclass[10pt,twocolumn]{article}

\usepackage{amsmath,amssymb}
\usepackage{amsthm}
\usepackage{times}
\usepackage{hyperref}
\usepackage{enumitem}
\usepackage{geometry}
\usepackage{tikz}
\usetikzlibrary{arrows.meta}

\theoremstyle{plain}
\newtheorem{theorem}{Theorem}[section]
\newtheorem{lemma}[theorem]{Lemma}

\theoremstyle{definition}
\newtheorem{definition}[theorem]{Definition}

\theoremstyle{remark}

\title{A Tight Lower Bound for Cycle Detection in Grid Graphs}
\author{Andrew Au\\[0.2em]
\small Independent Researcher\\[0.2em]
\small Corresponding author: \texttt{cshung@gmail.com}}
\date{\vspace{-1.5em}}

\begin{document}
\maketitle

\begin{abstract}
We prove that any algorithm for detecting cycles in an $m \times n$ grid graph, where cells are colored and adjacency is defined by matching colors, must read all $mn$ cells in the worst case for all grids with $m \geq 2$ and $n \geq 2$. The proof is by adversary argument: we construct an adaptive adversary that maintains ambiguity---one completion containing a cycle and one without---until the final cell is read. The construction proceeds by tiling the grid with $2 \times 2$, $2 \times 3$, $3 \times 2$, and $3 \times 3$ blocks, each equipped with an independent block adversary, composed via a checkerboard isolation scheme.
\end{abstract}


\section{Introduction}

Consider an $m \times n$ grid where each cell is assigned a character. Two adjacent cells (sharing an edge) with the same character are connected by an edge, forming an implicit graph. The \emph{cycle detection problem}~\cite{leetcode1559} asks whether this graph contains a cycle.

Solving this problem is straightforward---it reduces to standard undirected graph search. What is less obvious is the query complexity: how many cells must any correct algorithm inspect?

For arbitrary graphs, a simple lower bound argument shows that every edge must be examined, since a single edge can determine the presence or absence of a cycle. However, the grid graph is not arbitrary---it is implicitly defined by the grid coloring. Must an algorithm still read every cell?

We show the answer is yes, for all non-degenerate grids.

\begin{theorem}\label{thm:main}
For any grid of size $m \times n$ with $m \geq 2$ and $n \geq 2$, any deterministic algorithm that correctly decides whether the grid graph contains a cycle must read all $mn$ cells in the worst case.
\end{theorem}

The proof uses an adversary argument~\cite{yao77,blum-impagliazzo}. We construct an adaptive adversary that responds to the algorithm's queries while maintaining two consistent completions of the grid: one containing a cycle and one cycle-free. The adversary sustains this ambiguity until the very last cell is read.

The construction is inherently game-theoretic: the adversary and the algorithm play a sequential game in which the algorithm chooses which cell to query and the adversary chooses the response. The adversary's strategy is reactive---it adapts to the algorithm's choices to maximally delay resolution. Each base-case adversary is formalized as a finite state machine whose transitions depend on the algorithm's moves, making explicit the game tree underlying the lower bound.

\section{Preliminaries}

\begin{definition}[Grid Graph]
Given an $m \times n$ grid with a character assignment $c : [m] \times [n] \to \Sigma$, the \emph{grid graph} $G_c$ has vertex set $[m] \times [n]$ and edge set
\[
E = \{((i,j),(i',j')) : |i-i'| + |j-j'| = 1 \text{ and } c(i,j) = c(i',j')\}.
\]
\end{definition}

\begin{definition}[Adversary]
An adversary is an adaptive strategy that responds to each cell query with a character from $\Sigma$, such that after every prefix of queries, there exist two valid completions of the unread cells: one whose grid graph contains a cycle, and one whose grid graph is cycle-free.
\end{definition}

If such an adversary exists, no deterministic algorithm can decide cycle existence before reading all cells, establishing Theorem~\ref{thm:main}.

\section{Base Cases}

\subsection{The \texorpdfstring{$2 \times 2$}{2x2} Case}

\begin{lemma}\label{lem:2x2}
There exists an adversary for the $2 \times 2$ grid.
\end{lemma}

\begin{proof}
The adversary operates as a state machine with two states.

\textbf{State 0 (initial).} On each query: if at least two cells remain unread after this query, return \texttt{a} and remain in State~0. Otherwise (exactly one cell remains unread), return \texttt{a} and transition to State~1.

\textbf{State 1.} On the final query, the adversary chooses: return \texttt{a} to produce a cycle (all four cells are \texttt{a}), or \texttt{b} to produce no cycle (three \texttt{a}'s and one \texttt{b}).

This maintains ambiguity until the fourth and final read.
\end{proof}

\begin{figure}[ht]
\centering
\begin{tikzpicture}[scale=1.2]
  % Cycle completion
  \begin{scope}[xshift=0cm]
    \draw[step=1] (0,0) grid (2,2);
    \node at (0.5,1.5) {\texttt{a}};
    \node at (1.5,1.5) {\texttt{a}};
    \node at (0.5,0.5) {\texttt{a}};
    \node at (1.5,0.5) {\texttt{a}};
    \node[below] at (1,0) {\small Cycle};
  \end{scope}
  % No cycle
  \begin{scope}[xshift=3.5cm]
    \draw[step=1] (0,0) grid (2,2);
    \node at (0.5,1.5) {\texttt{a}};
    \node at (1.5,1.5) {\texttt{a}};
    \node at (0.5,0.5) {\texttt{a}};
    \node at (1.5,0.5) {\texttt{b}};
    \node[below] at (1,0) {\small No cycle};
  \end{scope}
\end{tikzpicture}
\caption{The $2 \times 2$ adversary's two completions. The last cell (bottom-right) determines the outcome.}
\label{fig:2x2}
\end{figure}
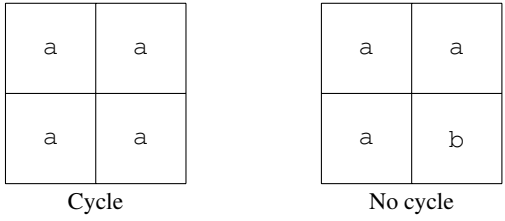

\subsection{The \texorpdfstring{$3 \times 3$}{3x3} Case}

\begin{lemma}\label{lem:3x3}
There exists an adversary for the $3 \times 3$ grid.
\end{lemma}

\begin{proof}
The adversary operates as a state machine with three states.

\textbf{State 0 (initial).} On each query: if the queried cell is not the center and fewer than eight boundary cells have been read, return \texttt{a} and remain in State~0. Otherwise (the queried cell is the center, or it is the eighth boundary cell read), return \texttt{b} and transition to State~1.

\textbf{State 1.} Return \texttt{a} for every query until exactly one cell remains unread, then transition to State~2.

\textbf{State 2.} On the final query, the adversary chooses: return \texttt{a} to produce a cycle, or \texttt{b} to produce no cycle.

\medskip

We verify that both completions are valid in each case.

\textbf{Case 1: The center triggers the transition to State~1.} The center received \texttt{b}. All read boundary cells are \texttt{a}, with one boundary cell remaining. In State~2, returning \texttt{a} completes the outer ring of \texttt{a}'s (creating a cycle); returning \texttt{b} breaks the ring (no cycle).

\textbf{Case 2: The eighth boundary cell triggers the transition to State~1.} Seven boundary cells received \texttt{a}; the eighth received \texttt{b}, breaking the outer ring. The center is the last cell read. In State~2, returning \texttt{a} creates a $2 \times 2$ cycle through the center; returning \texttt{b} prevents any cycle. The key observation is that no single boundary \texttt{b} can block all four $2 \times 2$ sub-squares containing the center, since each boundary cell belongs to at most two of the four sub-squares.
\end{proof}

\begin{figure}[ht]
\centering
\begin{tikzpicture}[scale=1.0]
  % Case 1: center = b, outer ring choice
  \begin{scope}[xshift=0cm]
    \draw[step=1] (0,0) grid (3,3);
    \node at (0.5,2.5) {\texttt{a}}; \node at (1.5,2.5) {\texttt{a}}; \node at (2.5,2.5) {\texttt{a}};
    \node at (0.5,1.5) {\texttt{a}}; \node at (1.5,1.5) {\texttt{b}}; \node at (2.5,1.5) {\texttt{a}};
    \node at (0.5,0.5) {\texttt{a}}; \node at (1.5,0.5) {\texttt{a}}; \node at (2.5,0.5) {\texttt{?}};
    \node[below] at (1.5,0) {\small Case 1};
  \end{scope}
  % Case 2: ring broken, center choice
  \begin{scope}[xshift=5cm]
    \draw[step=1] (0,0) grid (3,3);
    \node at (0.5,2.5) {\texttt{a}}; \node at (1.5,2.5) {\texttt{a}}; \node at (2.5,2.5) {\texttt{a}};
    \node at (0.5,1.5) {\texttt{a}}; \node at (1.5,1.5) {\texttt{?}}; \node at (2.5,1.5) {\texttt{a}};
    \node at (0.5,0.5) {\texttt{b}}; \node at (1.5,0.5) {\texttt{a}}; \node at (2.5,0.5) {\texttt{a}};
    \node[below] at (1.5,0) {\small Case 2};
  \end{scope}
\end{tikzpicture}
\caption{The $3 \times 3$ adversary at the final read (\texttt{?}). Case~1: the center is \texttt{b}, the last boundary cell decides the outer ring. Case~2: the ring is broken by one \texttt{b}, the center decides a $2 \times 2$ cycle.}
\label{fig:3x3}
\end{figure}
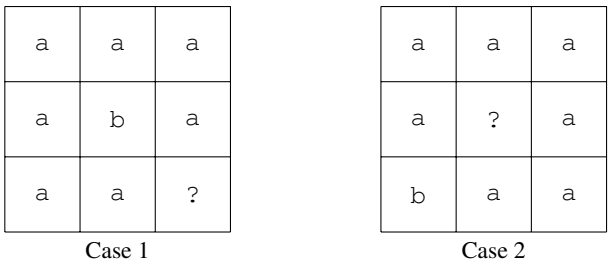

\subsection{The \texorpdfstring{$2 \times 3$}{2x3} Case}

\begin{lemma}\label{lem:2x3}
There exists an adversary for the $2 \times 3$ grid.
\end{lemma}

\begin{proof}
Label the six cells as follows:
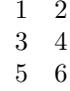
\[
\begin{array}{cc}
1 & 2 \\
3 & 4 \\
5 & 6
\end{array}
\]
The grid contains two $2 \times 2$ sub-blocks: the \emph{top block} $\{1,2,3,4\}$ and the \emph{bottom block} $\{3,4,5,6\}$. Define \emph{partner pairs}: cells $1$ and $2$ are partners, cells $5$ and $6$ are partners. Cells $3$ and $4$ have no partners.

The adversary operates as a state machine with three states.

\textbf{State 0 (initial).} On each query: if the queried cell has no partner, or its partner has not yet been read, return \texttt{a} and remain in State~0. Otherwise (the queried cell has a partner and that partner has already been read), return \texttt{b} and transition to State~1.

\textbf{State 1.} Return \texttt{a} for every query until exactly one cell remains unread, then transition to State~2.

\textbf{State 2.} On the final query, the adversary chooses: return \texttt{a} to produce a cycle, or \texttt{b} to produce no cycle.

\medskip

We verify that the transition to State~1 always leaves the adversary with a viable surviving block.

When a partner pair is completed (triggering the \texttt{b} and the transition), two facts hold: (A)~the other partner pair has not yet been completed, since every prior response was \texttt{a}---had the other pair completed first, it would have triggered the transition instead; and (B)~every cell read before this moment received \texttt{a}. The completed pair has one \texttt{a} and one \texttt{b}, breaking the $2 \times 2$ cycle in its block. In the other block, at most three of its four cells have been read, and all reads so far are \texttt{a}. This is precisely the $2 \times 2$ situation: the adversary continues returning \texttt{a} in State~1 until the last cell, then makes its choice in State~2.

In State~2, returning \texttt{a} completes this surviving block as all \texttt{a}'s (creating a $2 \times 2$ cycle); returning \texttt{b} breaks it (no cycle). This maintains ambiguity until the sixth and final read.
\end{proof}

\begin{figure}[ht]
\centering
\begin{tikzpicture}[scale=1.0]
  % Pair (1,2) completed, bottom block survives
  \begin{scope}[xshift=0cm]
    \draw[step=1] (0,0) grid (2,3);
    \node at (0.5,2.5) {\texttt{a}}; \node at (1.5,2.5) {\texttt{b}};
    \node at (0.5,1.5) {\texttt{a}}; \node at (1.5,1.5) {\texttt{a}};
    \node at (0.5,0.5) {\texttt{a}}; \node at (1.5,0.5) {\texttt{?}};
    \draw[thick,dashed,rounded corners] (-0.1,-0.1) rectangle (2.1,2.1);
    \node[below] at (1,0) {\small Top pair done};
  \end{scope}
  % Pair (5,6) completed, top block survives
  \begin{scope}[xshift=4.5cm]
    \draw[step=1] (0,0) grid (2,3);
    \node at (0.5,2.5) {\texttt{?}}; \node at (1.5,2.5) {\texttt{a}};
    \node at (0.5,1.5) {\texttt{a}}; \node at (1.5,1.5) {\texttt{a}};
    \node at (0.5,0.5) {\texttt{a}}; \node at (1.5,0.5) {\texttt{b}};
    \draw[thick,dashed,rounded corners] (-0.1,0.9) rectangle (2.1,3.1);
    \node[below] at (1,0) {\small Bottom pair done};
  \end{scope}
\end{tikzpicture}
\caption{The $2 \times 3$ adversary at the final read (\texttt{?}). When one partner pair is completed with a \texttt{b}, the surviving $2 \times 2$ block (dashed) reduces to the base case.}
\label{fig:2x3}
\end{figure}
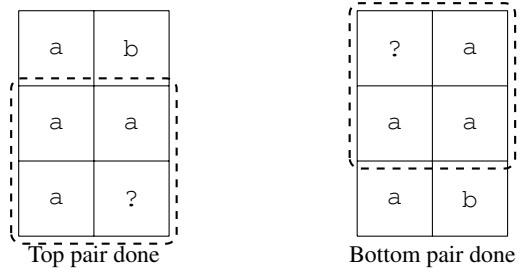

\section{Composition: Even Dimensions}

\begin{lemma}\label{lem:even}
For any $2p \times 2q$ grid with $p, q \geq 1$, there exists an adversary.
\end{lemma}

\begin{proof}
Partition the grid into a $p \times q$ array of $2 \times 2$ blocks. Color the blocks in a checkerboard pattern. Assign independent block adversaries as follows:
\begin{itemize}
    \item \textbf{Black blocks} use alphabet $\{\texttt{a}, \texttt{b}\}$ with the $2 \times 2$ adversary from Lemma~\ref{lem:2x2}.
    \item \textbf{White blocks} use alphabet $\{\texttt{c}, \texttt{d}\}$ with the same strategy.
\end{itemize}

Since adjacent blocks use disjoint alphabets, no edge crosses a block boundary. Each block's cycle is independent of every other block. An algorithm must fully resolve ambiguity within every block, requiring all cells of every block to be read.

The block adversaries coordinate only in one respect: all blocks report ``no cycle'' until the very last cell read across the entire grid. Whichever block contains the last unread cell makes the final decision.
\end{proof}

\begin{figure}[ht]
\centering
\begin{tikzpicture}[scale=0.8]
  \foreach \i in {0,1,2,3} {
    \foreach \j in {0,1,2} {
      \pgfmathtruncatemacro{\parity}{mod(\i+\j,2)}
      \ifnum\parity=0
        \fill[gray!25] (2*\i,2*\j) rectangle (2*\i+2,2*\j+2);
        \node at (2*\i+1,2*\j+1) {\small\texttt{a,b}};
      \else
        \fill[white] (2*\i,2*\j) rectangle (2*\i+2,2*\j+2);
        \node at (2*\i+1,2*\j+1) {\small\texttt{c,d}};
      \fi
    }
  }
  \draw[step=2] (0,0) grid (8,6);
\end{tikzpicture}
\caption{Checkerboard isolation for an $8 \times 6$ grid tiled with $2 \times 2$ blocks. Adjacent blocks use disjoint alphabets, preventing cross-block edges. Each block runs an independent adversary.}
\label{fig:checkerboard}
\end{figure}

\section{Composition: General Dimensions}

\begin{lemma}\label{lem:tiling}
Any $m \times n$ grid with $m \geq 2$ and $n \geq 2$ can be tiled by $2 \times 2$, $2 \times 3$, $3 \times 2$, and $3 \times 3$ blocks.
\end{lemma}

\begin{proof}
Write $m = 2a + 3b$ and $n = 2c + 3d$ for appropriate non-negative integers (possible for all $m, n \geq 2$). The grid decomposes into a product of row bands and column bands of widths 2 or 3, yielding blocks of the four required sizes.
\end{proof}

\begin{proof}[Proof of Theorem~\ref{thm:main}]
Tile the grid using Lemma~\ref{lem:tiling}. Apply the checkerboard isolation scheme from Lemma~\ref{lem:even}, generalized to blocks of mixed sizes: adjacent blocks use disjoint alphabets, preventing cross-block edges. Equip each block with its corresponding adversary (Lemmas~\ref{lem:2x2},~\ref{lem:3x3},~\ref{lem:2x3}, and the transposed variant). The block adversaries coordinate so that ambiguity is maintained across the entire grid until the last cell is read.
\end{proof}

\section{Conclusion}

We have shown that cycle detection on grid graphs requires reading every cell, matching the trivial upper bound. The proof relies on a tiling-based adversary construction that decomposes the grid into independent blocks, each maintaining local ambiguity. The checkerboard alphabet isolation ensures no information leaks between blocks.

The construction is game-theoretic in nature: each block adversary plays a sequential game against the algorithm, adapting its responses to the algorithm's query order. The state machine formulation makes this game explicit---the adversary's strategy is not a fixed worst-case input but a reactive policy that counters any possible algorithm. This adaptive quality is essential: no single static grid coloring forces all algorithms to read every cell, but the adversary can force \emph{each} algorithm to do so by tailoring its responses.

The construction uses an alphabet of size four: two colors per block ($\{\texttt{a}, \texttt{b}\}$ and $\{\texttt{c}, \texttt{d}\}$) to isolate adjacent blocks. Whether the lower bound can be established with fewer colors remains open.

\section*{Acknowledgements}

This research did not receive any specific grant from funding agencies in the public, commercial, or not-for-profit sectors.


\begin{thebibliography}{9}

\bibitem{yao77}
A.~C.~Yao,
\textit{Probabilistic computations: Toward a unified measure of complexity},
in Proceedings of the 18th IEEE Symposium on Foundations of Computer Science (FOCS), pp.~222--227, 1977.

\bibitem{blum-impagliazzo}
M.~Blum and R.~Impagliazzo,
\textit{Generic oracles and oracle classes},
in Proceedings of the 28th IEEE Symposium on Foundations of Computer Science (FOCS), pp.~118--126, 1987.

\bibitem{leetcode1559}
LeetCode,
\textit{Problem 1559: Detect Cycles in 2D Grid},
\url{https://leetcode.com/problems/detect-cycles-in-2d-grid/}.

\end{thebibliography}
\end{document}